\documentclass[runningheads]{llncs}

\makeatletter
\usepackage[bookmarks,unicode,colorlinks=true]{hyperref}
   \def\@citecolor{blue}
   \def\@urlcolor{blue}
   \def\@linkcolor{blue}

\def\orcidID#1{\smash{\href{http://orcid.org/#1}{\protect\raisebox{-1.25pt}{\protect\includegraphics{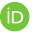}}}}}
\makeatother
\usepackage{wrapfig,lipsum,booktabs}

\usepackage{amssymb,amsmath,amsfonts}
\usepackage{temporal}
\usepackage{fontawesome} 
\usepackage{pifont}
\usepackage{textcomp}
\usepackage[mathscr]{eucal}
\usepackage{enumitem}
\usepackage{svg}

\usepackage{tikz}

\usetikzlibrary{automata,positioning,decorations.markings,arrows,intersections,shapes,
  shapes.symbols,calc,backgrounds,shadings}
\usetikzlibrary{overlay-beamer-styles}

\tikzset{
  >=stealth,
  box state/.style={draw,rectangle,rounded corners,fill=safecellcolor,minimum size=8mm},
  prob state/.style={draw,very thick,shape=circle,darkblue,minimum size=3mm,inner sep=0mm},
  node distance=2cm,on grid,auto, initial text=,
  every loop/.style={shorten >=0pt},
  accepting/.style={double distance=1.2pt, outer sep = 0.6pt+\pgflinewidth},
  accepting dot/.style={above=-2.5pt,circle,fill,darkgreen,inner sep=2pt,radius=1pt},
  loop above/.append style={every loop/.append style={out=120, in=60, looseness=6}},
  loop below/.append style={every loop/.append style={out=300, in=240, looseness=6}},
  loop left/.append style={every loop/.append style={out=210, in=150, looseness=6}},
  loop right/.append style={every loop/.append style={out=30, in=330, looseness=6}},
  accepting arc/.style={dashed},
  marked/.style={
    dashed,
    opacity=0.3
  },
  marked on/.style={alt=#1{marked}{}},
}

\definecolor{darkgreen}{rgb}{0,0.6,0}
\definecolor{lightblue}{rgb}{0.5,0.6,1.0}
\definecolor{lightgray}{rgb}{0.98,0.98,0.98}
\definecolor{mauve}{rgb}{0.58,0,0.82}
\definecolor{sienna}{rgb}{0.6,0.18,0.09}
\colorlet{darkblue}{blue!60!black}
\colorlet{darkred}{red!50!black}
\colorlet{safecellcolor}{yellow!5}
\colorlet{goodcellcolor}{green!10}
\colorlet{badcellcolor}{blue!10}

\DeclareMathOperator{\infi}{inf}

\newcommand{\tuple}[1]{( #1 )}

\newcommand{\bB}{\mathbb{B}}
\newcommand{\DIST}{{\mathcal D}}

\newcommand{\Aa}{\mathcal{A}}
\newcommand{\Mm}{\mathcal{M}}
\newcommand{\Gg}{\mathcal{G}}
\newcommand{\Ff}{\mathcal{F}}

\newcommand{\FRuns}{\mathit{FRuns}}
\newcommand{\Runs}{\mathit{Runs}}
\DeclareMathOperator{\last}{\mathit{last}}
\newcommand{\set}[1]{\left\{ #1 \right\}}

\newcommand{\seq}[1]{\langle #1 \rangle}
\DeclareMathOperator{\supp}{\mathit{supp}}
\newcommand{\eE}{\mathbb E}
\newcommand{\Real}{\mathbb R}

\DeclareMathOperator{\PSat}{\mathsf{PSyn}}
\DeclareMathOperator{\PSemSat}{\mathsf{PSem}}

\usepackage{xcolor,colortbl}

\definecolor{Gray}{gray}{0.85}
\definecolor{LightCyan}{rgb}{0.88,1,1}

\let\llncssubparagraph\subparagraph
\let\subparagraph\paragraph
\usepackage[compact]{titlesec}
\let\subparagraph\llncssubparagraph

\usepackage{mathtools}

\newcommand{\mMAX}{\mathrm{Max}}
\newcommand{\mMIN}{\mathrm{Min}}
\newcommand{\pto}{\xrightharpoondown{}}
\newcommand{\hSigma}{\overline{\Sigma}}
\newcommand{\hPi}{\overline{\Pi}}

\hypersetup{pdfauthor={Name}}
\begin{document}
\title{Alternating Good-for-MDP Automata
\thanks{
\includegraphics[height=8pt]{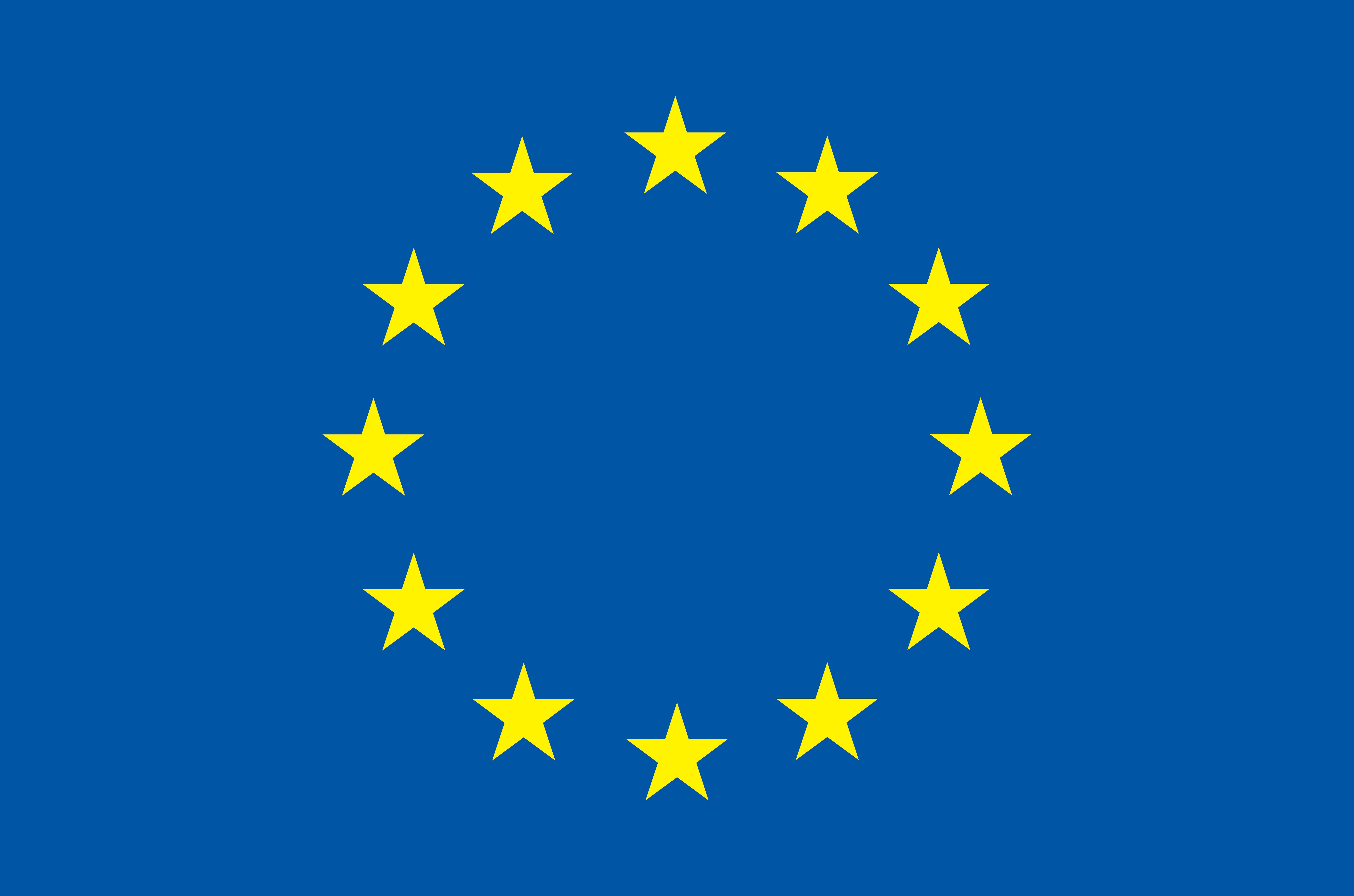} This project has received funding from the European Union’s Horizon 2020 research and innovation programme under grant agreements 864075 (CAESAR), and 956123 (FOCETA).
This work is supported in part by the National Science Foundation grant 2009022, by a CU Boulder Research and Innovation Office grant, and
by the EPSRC through grant EP/V026887/1.}
}
\author{Ernst Moritz Hahn\inst{1}\orcidID{0000-0002-9348-7684}
  \and Mateo Perez\inst{2}\orcidID{0000-0003-4220-3212}
  \and Sven Schewe\inst{3}\orcidID{0000-0002-9093-9518}
  \and Fabio Somenzi\inst{2}\orcidID{0000-0002-2085-2003}
  \and Ashutosh Trivedi\inst{2}\orcidID{0000-0001-9346-0126}
  \and Dominik Wojtczak\inst{3}\orcidID{0000-0001-5560-0546}}

\institute{
  University of Twente, The Netherlands
  \and
  University of Colorado Boulder, USA
  \and
  University of Liverpool, UK
}
\authorrunning{E. M. Hahn, M. Perez, S. Schewe, F. Somenzi, A. Trivedi, and D. Wojtczak}
\titlerunning{Alternating GFM Automata}

\maketitle              

\begin{abstract}
When omega-regular objectives were first proposed in model-free reinforcement learning (RL) for controlling MDPs, deterministic Rabin automata were used in an attempt to provide a direct translation from their transitions to scalar values.
While these translations failed, it has turned out that it is possible to repair them by using good-for-MDPs (GFM) B\"uchi automata instead.
These are nondeterministic B\"uchi automata with a restricted type of nondeterminism, albeit not as restricted as in good-for-games automata.
Indeed, deterministic Rabin automata have a pretty straightforward translation to such GFM automata, which is bi-linear in the number of states and pairs.
Interestingly, the same cannot be said for deterministic Streett automata:
a translation to nondeterministic Rabin or B\"uchi automata comes at an exponential cost, even without requiring the target automaton to be good-for-MDPs.
Do we have to pay more than that to obtain a good-for-MDP automaton?
The surprising answer is that we have to pay significantly less when we instead expand the good-for-MDP property to alternating automata:
like the nondeterministic GFM automata obtained from deterministic Rabin automata, the alternating good-for-MDP automata we produce from deterministic Streett automata are bi-linear in the the size of the deterministic automaton and its index, and can therefore be exponentially more succinct than minimal nondeterministic B\"uchi automata.

\end{abstract}

\section{Introduction}
\label{sec:intro}

Omega-automata \cite{Thomas90b,Perrin04} have found renewed interest---often as the result of translating a formula in LTL~\cite{Pnueli77}---as specifications of qualitative objectives in reinforcement learning (RL)~ \cite{Sutton18}. 
The acceptance condition of an $\omega$-automaton determines the reward whose cumulative return the learning agent strives to maximise.
The relation between the automaton and the reward signal should ensure that a strategy that maximises the expected return also maximises the probability to realise the objective.
This so-called \emph{faithfulness} requirement~\cite{Hahn20c} restricts the type of $\omega$-automaton that can be used to represent the objective, and this paper concerns how to find the right type of $\omega$-automata.

Deterministic automata with various types of acceptance conditions have been used in model checking and strategy synthesis \cite{Baier08}; notably, B\"uchi, parity, Rabin, and Streett.  While deterministic B\"uchi automata---including generalised deterministic B\"uchi automata---do not accept all $\omega$-regular languages, deterministic parity, Rabin, and Streett automata do; therefore, they are employed in the formulation of general solutions to synthesis problems.
In addition, maximising the chance of meeting a parity and Rabin winning conditions in a game can be obtained using positional strategies, while Streett winning conditions require finite additional memory.
This means that a positional strategy for a Markov decision process (MDP) or a Markov game endowed with a parity or Rabin objective can be turned into a strategy to control the environment, such that strategy only uses the state of the $\omega$-automaton for the objective as memory.
Strategy computation methods for both Rabin and Streett automata have been studied extensively \cite{Buhrke96,deAlfa98,Piterm06c}.
These methods, however, are not applicable in RL.
In order to apply RL to the computation of optimal strategies for $\omega$-regular objectives, we have to devise a scheme for doling out rewards that, for generality, depend on the given $\omega$-automaton, and perhaps on some hyperparameters of the learning algorithm, but not on the MDP---or the Markov game---for which a control strategy is sought.

Two features of RL algorithms significantly affect the choice of translation from acceptance condition to rewards:
1) they require positional optimal strategies after the translation to rewards, as they learn values of states and transitions, and 2) the same transitions will always be optimal once a property is translated into scalar rewards.
This appears to effectively exclude using Streett conditions, as optimal control requires memory for Streett objectives.
This is regrettable, as Streett objectives do occur in practice.
GR(1) \cite{BloemJPPS12} conditions, for example, translate smoothly into Streett objectives (in the pure original form, to one pair Streett objectives), such that a conjunction of GR(1) objectives will always have a natural representation as a deterministic Streett automaton.
Likewise, each strong fairness requirement produces a Streett pair.
Moreover, minimising the chance of satisfying a Rabin condition given as a deterministic Rabin automaton (DRA) is also equivalent to maximising the chance of satisfying the Streett condition given by its dual.

A natural way to move to simpler acceptance conditions requires some form of nondeterminism.
Full recourse to nondeterminism, however, is not compatible with the computation of optimal strategies for MDPs or Markov games.
If we want to move away from using deterministic automata to describe the objective, we therefore need to impose restrictions on the automaton's nondeterminism.
The precise nature of these restrictions depends on the type of environment interacting with the agent, whose control strategy we want to build.
If the environment is a Markov decision process \cite{Put94}, the automaton needs to be Good-for-MDPs (GFM) \cite{Hahn20,Vardi85,Courco95}, while, for Markov games, with two strategic players, the stronger requirements of Good-for-Games automata \cite{Henzin06} must be satisfied. 
GFM automata have the advantage that they can use simpler acceptance mechanisms.
In particular, the GFM automata developed so far are nondeterministic B\"uchi automata, and being able to use a simple acceptance mechanism like B\"uchi is quite beneficial for RL \cite{Hahn19,Hahn20}---though it is possible to use parity automata, using them comes at a cost \cite{Hahn20b}.

When starting with a deterministic Streett automaton (DSA), a translation to a nondeterministic B\"uchi automaton (NBA) \cite{SV89}, or even to a nondeterministic Rabin automaton \cite{SV89}, comes at the cost of an exponential blow-up, even without the restriction to GFM automata.
This raises the question of whether or not there is a different way to efficiently translate the DSA into a suitable automaton.
Our main result is that
\textbf{alternating GFM automata can be exponentially more succinct than general nondeterministic B\"uchi automata}.

Intuitively, this should not be possible:
the reason for the exponential blow-up from DSAs to NBAs (and even NRAs) is that one will either need some form of memory, such as a latest appearance record (LAR), or a nondeterministic guess as to which way each Streett pair is satisfied in.
Recall that a Streett pair consists of a green and a red set of states or transitions, and it is satisfied if no entry of the red set \emph{or} some entry of the green set occurs infinitely often; a nondeterministic automaton can guess, for each pair, to validate either of these conditions.
The intuitive effect of this blow-up would be that starting with DSAs is something that efficient RL approaches will struggle with.

The key message of this paper is that one can trade-off \emph{alternation} for memory in computing an optimal strategy by moving to \emph{alternating} GFM automata instead of traditional nondeterministic ones. 
Here is the interesting bit: while we do, unsurprisingly, need a latest appearance record (LAR \cite{Gurevi82,Dziemb97}) in the control strategy we develop, the blow-up due to the LAR is not necessary while learning the optimal strategy!

\section{Preliminaries}

A \emph{probability distribution} over a finite set $S$ is a function
$d \colon S {\to} [0, 1]$ such that $\sum_{s \in S} d(s) = 1$.  Let
$\DIST(S)$ denote the set of all discrete distributions over $S$.  We
say a distribution ${d \in \DIST(S)}$ is a \emph{point distribution}
if $d(s) {=} 1$ for some $s \in S$.  For $d \in \DIST(S)$ we write
$\supp(d)$ for $\set{s \in S \colon d(s) > 0}$.

\subsection{Stochastic Game Arenas and Markov Decision Processes}
\label{sec:MDPs}

A \emph{stochastic game arena} $\Gg$ is a tuple
$(S, s_0, A, T, S_\mMAX, S_\mMIN, AP, L)$, where
   $S$ is a finite set of states,
   $s_0 \in S$ is the initial state,
   $A$ is a finite set of {\it actions},
   $T : S \times A \pto \DIST(S)$ is the \emph{probabilistic
  transition (partial) function},
  $\{S_\mMAX, S_\mMIN\}$ is a
partition of the set of states $S$, 
$AP$ is the set of {\it 
  atomic propositions}, and 
  $L : S \to 2^{AP}$ is the {\it labeling
  function}.
For $s \in S$, $A(s)$ denotes the set of actions enabled in $s$.  
For states $s, s' \in S$ and $a \in A(s)$
we write $p(s' | s, a)$ for $T(s, a)(s')$.  

A {\it run} of $\Gg$ is an
$\omega$-word
$\seq{s_0, a_1, s_1, \ldots} \in S \times (A \times S)^\omega$
such
that $p(s_{i+1} | s_{i}, a_{i+1}) > 0$ for all $i \geq 0$.  A finite
run is a finite such sequence, that is, a word in
$ S \times (A \times S)^*$.  
For a {\it run} $r = \seq{s_0, a_1, s_1, \ldots}$ we define the corresponding
labeled run as $L(r) = \seq{L(s_0), L(s_1), \ldots} \in (2^{AP})^\omega$.
We write $\Runs^\Gg (\FRuns^\Gg)$ for the set of runs
(finite runs) of the SGA $\Gg$ and $\Runs{}^\Gg(s) (\FRuns{}^\Gg(s))$
for the set of runs (finite runs) of the SGA $\Gg$ starting from state
$s$.  We write $\last(r)$ for the last state of a finite run $r$.

A game on an SGA $\Gg$ is played between two players, Max and Min, by moving a token through the states of the arena. 
The game begins with a token in an {\em initial state} $s_0$; players $\mMAX$ and $\mMIN$ construct an infinite run by taking turns to choose enabled actions when the token is in a state controlled by them, and then moving the token to a successor state sampled from the selected distribution.
A strategy of player $\mMIN$ in $\Gg$ is a partial function $\pi \colon \FRuns \pto \DIST(A)$, defined for $r \in \FRuns$ if and only if $\last(r) \in S_\mMIN$, such that $\supp(\sigma(r)) \subseteq A(\last(r))$. A strategy $\sigma$ of player Max is defined analogously.
We drop the subscript $\Gg$ when the arena is clear from the context.
Let $\Sigma_\Gg$ and $\Pi_\Gg$ be the sets of all strategies of player Max and player Min, respectively.

A memory skeleton for $\Gg$ is a tuple $\mathbf{M} = (M, m_0, \alpha_u)$ where $M$ is a finite set of memory states, $m_0 \in M$ is the initial state, and $\alpha_u: M {\times} 2^{AP} \to M$ is the memory update function. 
The extended memory update $\hat{\alpha}_u \colon M {\times} (2^{AP})^* \to M$ 
can be defined in the usual manner. 
A finite memory strategy of player $\mMAX$ in $\Gg$ over a memory skeleton $\mathbf{M}$ is a Mealy machine 
$(\mathbf{M}, \alpha_x)$ where $\alpha_x: S_\mMAX {\times} M \to \DIST(A)$ is the {\it next action function} 
that suggests the next action based on the SGA and the memory state. 
The semantics of a finite memory strategy $(\mathbf{M}, \alpha_x)$ is given as a strategy 
$\sigma \in \Sigma_\Gg$ such that for every $r \in \FRuns$ with $\last(r) \in S_\mMAX$, we have that 
$\sigma(r) = \alpha_x(\last(r), \hat{\alpha}_u(m_0, L(r)))$.

A strategy $\sigma$ is {\it pure} if $\sigma(r)$ is a point distribution wherever it is defined; otherwise, $\sigma$ is \emph{mixed}.  We say that $\sigma$ is {\it stationary} if $\last(r) = \last(r')$ implies $\sigma(r) = \sigma(r')$ wherever $\sigma$ is defined.  
A strategy is \emph{positional} if it is both pure and stationary. 
We write $\hSigma_\Gg$ and $\hPi_\Gg$ for the sets of all \emph{positional} strategies of player Max and player Min, respectively.

Let $\Runs^\Gg_{\sigma,\pi}(s)$ denote the subset of runs $\Runs^\Gg(s)$
starting from state $s$ that are consistent with player Max and player Min following strategies
$\sigma$ and $\pi$, respectively.
The behavior of an SGA $\Gg$ under a strategy pair $(\sigma, \pi) \in \Sigma_\Gg
\times \Pi_\Gg$ is defined on the probability space $(\Runs^\Gg_{\sigma, \pi}(s), \Ff_{\Runs^\Gg_{\sigma, \pi}(s)}, \Pr^\Gg_{\sigma, \pi}(s))$ over
the set of infinite runs $\Runs^\Gg_{\sigma,\pi}(s)$, where $\Ff_{\Runs^\Gg_{\sigma, \pi}(s)}$ is the standard $\sigma$-algebra over them.  Given a
random variable $f \colon \Runs^\Gg \to \Real$ over the infinite runs
of $\Gg$, we denote by $\eE^\Gg_{\sigma, \pi}(s) \set{f}$ the expectation of $f$ over the runs in the probability space $(\Runs^\Gg_{\sigma, \pi}(s), \Ff_{\Runs^\Gg_{\sigma, \pi}(s)},
\Pr^\Gg_{\sigma, \pi}(s))$.

We say that a SGA is a Markov decision process if $A(s)$ is a singleton for every $s \in S_\mMIN$ and is a Markov chain if $A(s)$ is singleton for every $s \in S$.
To distinguish an MDP from an  SGA, we denote an MDP by $\Mm$ and write its signature $\Mm = (S, s_0, A, T, AP, L)$ by assigning the (choiceless) states of player Min to player Max.
The notions defined for SGAs naturally carry over to MDPs.

\subsection{Omega-Automata}
\label{sec:omega-automata}

An alphabet is a finite set of letters.
We write $\bB$ for the binary alphabet $\set{0, 1}$.
A finite word over an alphabet $\Sigma$ is a finite concatenation of symbols from $\Sigma$. 
Similarly, an $\omega$-word $w$ over $\Sigma$ is a function
$w : \omega \to \Sigma$ from the natural numbers to $\Sigma$.
We write $\Sigma^*$ and $\Sigma^\omega$ for the set of finite and $\omega$-strings over $\Sigma$. 

  An $\omega$-automaton $\mathcal{A} =
 \tuple{\Sigma,Q,q_0,\delta,\alpha}$ consists of a finite alphabet
  $\Sigma$, a finite set of states $Q$, an initial state $q_0 \in Q$,
  a transition function $\delta : Q \times \Sigma \to 2^Q$, and an
  acceptance condition $\alpha: Q^\omega \to \bB$.  A
  \emph{deterministic} automaton is such that $\delta(q,\sigma)$ is a
  singleton for every state $q$ and alphabet letter $\sigma$.  For
  deterministic automata, we write $\delta(q,\sigma) = q'$ instead of
  $\delta(q,\sigma) = \{q'\}$.
  
  A \emph{run} of an automaton
  $\mathcal{A} = \tuple{\Sigma,Q,q_0,\delta,\alpha}$ on word
  $w \in \Sigma^\omega$ is a function $\rho : \omega \to Q$ such that
  $\rho(0) = q_0$ and $\rho(i+1) \in \delta(\rho(i), w(i))$.  A run
  $\rho$ is \emph{accepting} if $\alpha(\rho) = 1$.  
  A word $w$ is
  accepted by $\mathcal{A}$ if there exists an accepting run of
  $\mathcal{A}$ on $w$.  The language of $\mathcal{A}$, written
  $\mathcal{L}(\mathcal{A})$, is the set of words accepted by $\mathcal{A}$.
  The set of states that appear infinitely often in $\rho$ is written
  $\infi(\rho)$.
  A deterministic automaton $\mathcal{D}$ has exactly
  one run for each word in $\Sigma^\omega$.
  We write
  $\infi^{\mathcal{D}}(w)$ for the set of states that appear
  infinitely often in the unique run of $\mathcal{D}$ on $w$;  
  when clear from the context, 
  we drop the superscript and simply write $\infi(w)$.

Several ways to give finite presentations of the acceptance
conditions are in use. The ones relevant to this paper are listed below.
\begin{itemize}
\item A \emph{B\"uchi} acceptance condition is specified by a set of
  states $F \subseteq Q$ such that 
  \[
  \alpha(\rho) = [\,\infi(\rho) \cap F \neq \emptyset\,].
  \]
\item A \emph{Rabin} acceptance condition of index $k$ is specified by
  $k$ pairs of sets of states, $\{\tuple{R_i,G_i}\}_{1 \leq i \leq k}$, and 
  intuitively a run should visit at least one set of Red (ruinous) states 
  finitely often and its corresponding Green (good) set of states infinitely often.
  Formally, 
  \[
  \alpha(\rho) = [\,\exists\, 1 \leq i \leq
  k \scope \infi(\rho) \cap R_i = \emptyset \text{ and }
  \infi(\rho) \cap G_i \neq \emptyset\,].
  \]
\item A \emph{Streett} acceptance condition of index $k$ is specified by
  $k$ pairs of sets of states, $\{\tuple{G_i,R_i}\}_{1 \leq i \leq k}$, and 
  intuitively a run should visit each Red set of states 
  finitely often or its corresponding Green set of states infinitely often.
  Formally, 
  \[
  \alpha(\rho) = [\,\forall\, 1 \leq i \leq
  k \scope \infi(\rho) \cap R_i = \emptyset \text{ or }
  \infi(\rho) \cap G_i \neq \emptyset\,].
  \]
\end{itemize}

We also allow for moving the acceptance condition from states to transitions.
For a B\"uchi acceptance condition, this means defining a set $F \subseteq Q\times \Sigma \times Q$ of final transitions, where the $i^{th}$ transition for a run $\rho$ of a word $w$ is $t(i) = \big(\rho(i),w(i),\rho(i+1)\big)$, $\inf_t(\rho,w)$ is the set of transitions that occur infinitely often, and a run $\rho$ is accepting for $w$ if $\inf_t(\rho,w)\cap F \neq \emptyset$.

\subsection{Semantic Satisfaction: Optimal Strategies Against $\omega$-Automata}
\label{ssec:semanticProb}

Given an MDP $\Mm = (S, s_0, A, T, AP, L)$ and an $\omega$-automaton $\Aa = \tuple{2^{AP},Q,q_0,\delta,\alpha}$, we are interested in strategies that maximise the probability that the labels of a run of $\Mm$ form an $\omega$-word in the language of $\Aa$. A strategy $\sigma \in \Sigma_\Mm$ and initial state $s \in S$ determine a sequence $X_i$ of random variables
denoting the $i^{th}$ state of the MDP, where $X_0 = s$.

We define the optimal satisfaction probability $\PSemSat^\Mm_{\Aa}(s)$ as 
\[
\PSemSat^\Mm_{\Aa}(s) 
= 
\sup_{\sigma \in \Sigma_{\Mm}} \Pr{}^\Mm_{\sigma}(s) \set{ \seq{L(X_0), L(X_1), \ldots } \in \mathcal{L}(\mathcal{A})} .
\]
We say that a strategy $\sigma \in \Sigma_\Mm$ is optimal for $\Aa$ if 
\[
\Pr{}^\Mm_{\sigma}(s) \set{ \seq{L(X_0), L(X_1), \ldots } \in \mathcal{L}(\mathcal{A})} = \PSemSat^\Mm_{\Aa}(s)
\]
for all $s \in S$.

\subsection{Good-for-MDPs Automata}
\label{ssec:product}
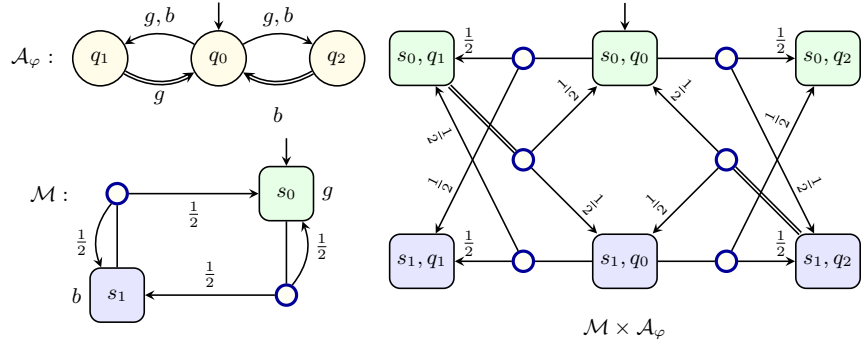
\begin{figure}[t]
  \centering
  \begin{tikzpicture}[scale=0.9,transform shape,
      every state/.append style={fill=yellow!10},semithick,
      el/.style = {inner sep=2pt, align=left, sloped}]
    \begin{scope}
      \node[state, initial above] (q0) {$q_0$};
      \node[state] (q1) [left=1.75cm of q0] {$q_1$};
      \node [left=1cm of q1] {$\Aa_\varphi:$};
      \node[state] (q2) [right=1.75cm of q0] {$q_2$};
      \path[->]
      (q0) edge [bend right] node[above] {$g, b$} (q1)
      (q1) edge [double,bend right] node[label={below:$g$}] {} (q0)
      (q0) edge [bend left]  node[above] {$g,b$} (q2)
      (q2) edge [double,bend left]  node[label={below:$b$}] {} (q0);
    \end{scope}
    \begin{scope}[yshift=-2cm,xshift=1cm]
      \node[box state,initial above, fill=goodcellcolor,label=right:$g$] (S0) {$s_0$};
      \node[prob state,fill=safecellcolor] (Prob) [below=1.5cm of S0] {};
      \node[box state,fill=badcellcolor,label=left:$b$] (S1) [left=2.5cm of Prob] {$s_1$};
      \node[prob state,fill=safecellcolor] (Prob2) [above=1.5cm of S1] {};
      \node [left=1cm of Prob2] {$\Mm:$};
      
      \path[-]
      (S0) edge[swap] node[label={left:$ $}] {} (Prob);
      \path[->]
      (Prob) edge[swap] node {$\frac{1}{2}$} (S1)
      edge[bend right] node[label={right:$\frac{1}{2}$}]{} (S0);
      \path[-]
      (S1) edge[swap] node[label={left:$ $}] {} (Prob2);
      \path[->]
      (Prob2) edge[swap] node {$\frac{1}{2}$} (S0)
      edge[bend right] node[label={left:$\frac{1}{2}$}]{} (S1);
    \end{scope}
    \begin{scope}[xshift=6cm]
      \node[box state,initial above, fill=goodcellcolor,label=right:$ $] (S0){$s_0, q_0$};
          \node[box state,fill=badcellcolor] (S1)[below=3cm of S0] {$s_1, q_0$};
          \node[box state,fill=goodcellcolor] (S2)[left=3cm of S0] {$s_0, q_1$};
          \node[box state,fill=badcellcolor] (S3)[below=3cm of S2] {$s_1, q_1$};
          \node[box state,fill=goodcellcolor] (S4)[right=3cm of S0] {$s_0, q_2$};
          \node[box state,fill=badcellcolor] (S5)[below=3cm of S4] {$s_1,
            q_2$};
          \node [below=1cm of S1] {$\Mm\times\Aa_\varphi$};

          \node[prob state,fill=safecellcolor] (Prob1) [left=1.5cm of S0] {};
          \node[prob state,fill=safecellcolor] (Prob2) [left=1.5cm of S1] {};
          \node[prob state,fill=safecellcolor] (Prob3) [right=1.5cm of S0] {};
          \node[prob state,fill=safecellcolor] (Prob4) [right=1.5cm of S1] {};
          \node[prob state,fill=safecellcolor] (Prob5) [below=1.5cm of Prob1] {};
          \node[prob state,fill=safecellcolor] (Prob6) [below=1.5cm of Prob3] {};
          \path[-]
          (S0) edge[swap] node[label={left:$ $}] {} (Prob1)
          (S0) edge[swap] node[label={left:$ $}] {} (Prob3)
          (S1) edge[swap] node[label={left:$ $}] {} (Prob2)
          (S1) edge[swap] node[label={left:$ $}] {} (Prob4)
          (S2) edge[double,swap] node[label={left:$ $}] {} (Prob5)
          (S5) edge[double,swap] node[label={left:$ $}] {} (Prob6);

          \path[->] (Prob1) edge[swap] node[el,above,pos=0.75] {$\frac{1}{2}$} (S2)
          edge[swap] node[el,above,pos=0.75] {$\frac{1}{2}$} (S3);
          \path[->] (Prob2) edge[swap] node[el,below,pos=0.75]
               {$\frac{1}{2}$}  (S2)  edge[swap] node[el,above,pos=0.75] {$\frac{1}{2}$}(S3);
          \path[->] (Prob3) edge[swap] node[el,above,pos=0.75] {$\frac{1}{2}$}  (S4) edge[swap] node[el,above,pos=0.75] {$\frac{1}{2}$} (S5);
          \path[->] (Prob4) edge[swap]  node[el,above,pos=0.75] {$\frac{1}{2}$}  (S4)
          edge[swap]  node[el,above,pos=0.75] {$\frac{1}{2}$} (S5);
          \path[->] (Prob5) edge[swap]  node[el,above,pos=0.75] {$\frac{1}{2}$}
          (S0) edge[swap]  node[el,above,pos=0.75] {$\frac{1}{2}$} (S1);
          \path[->] (Prob6) edge[swap]  node[el,above,pos=0.75] {$\frac{1}{2}$}  (S0) edge[swap] node[el,above,pos=0.75] {$\frac{1}{2}$} (S1);
    \end{scope}
  \end{tikzpicture}
  \caption{\small {\bf Syntactic and Semantic Probabilities Differ. }
    A nondeterministic B\"uchi
    automaton $\Aa_\varphi$ (top left) that recognises the language of infinitely many 
    $g$'s or infinitely many $b$'s (which accepts all $\omega$-words) and a Markov
    decision process $\Mm$, whose set of actions is a singleton (a Markov chain).
    Note that double edges mark accepting transitions.
    Notice that the MDP $\Mm$ satisfies the property with probability $1$. Their
    product is shown on the right side, where there is no accepting end-component. Hence, the
    probability of reaching the accepting end-component (under any strategy) is $0$.} 
  \label{fig:NBW}
\end{figure}

Given an MDP $\Mm = (S, s_0, A, T, AP, L)$ and automaton $\Aa =  \tuple{2^{AP},Q,q_0,\delta,\alpha}$, the \emph{probabilistic model checking} problem is to find the optimal value $\PSemSat^\Mm_{\Aa}(s)$ and an optimal strategy in $\Sigma_\Mm$.
An intuitive way to compute $\PSemSat^\Mm_{\Aa}(s)$ is to build the synchronous product $\Mm{\times}\Aa$ of $\Mm$ and $\Aa$ and compute the optimal probability $\PSat^\Mm_{\Aa}(s)$ and a strategy that maximizes the probability of satisfying the acceptance condition. 
If these values $\PSemSat^\Mm_{\Aa}(s)$ and $\PSat^\Mm_{\Aa}(s)$ coincide for all possible $\Mm$, then the automaton is said to be \emph{good-for-MDPs} \cite{Hahn20}.

The synchronous product is an MDP $\Mm \times \Aa = \tuple{S\times Q, (s_0, q_0), A\times Q, T^\times, AP, L}$, where
$$T^\times((s,q),(a,q'))(s',q'') = \begin{cases}
    T(s,a)(s') & \text{if } q' \in \delta(q, L(s)) \text{ and } q'' = q' \\
    0 & \text{if } q' \in \delta(q, L(s)) \text{ and } q'' \neq q' \\
    \text{undefined} & \text{otherwise .}
\end{cases}$$
A strategy $\sigma \in \Sigma_{\Mm\times \Aa}$ and initial state $s \in S$ determine a sequence $(X_i, Q_i)$ of random variables
denoting the $i^{th}$ state of the product MDP, where $X_0 = s$ and $Q_0 = q_0$.
The syntactic probability is defined to be
\[
    \PSat^\Mm_{\Aa}(s) = \sup_{\sigma \in \Sigma_{\Mm\times \Aa}} \eE^{\Mm\times \Aa}_{\sigma}(s) \set{ \alpha( \seq{Q_0, Q_1, \ldots })}.
\]
An automaton is good-for-MDPs if $\PSemSat^\Mm_{\Aa}(s) = \PSat^\Mm_{\Aa}(s)$ for all MDPs $\Mm$ and states $s \in S$.
Figure~\ref{fig:NBW} shows an example of an $\omega$-automaton with B\"uchi acceptance condition that is not GFM, while Figure~\ref{fig:DBW} shows an automaton that is a GFM (since every deterministic $\omega$-automaton is GFM.).

\begin{figure}[t]
   \centering
  \begin{tikzpicture}[scale=0.9,transform shape,
    every state/.append style={fill=yellow!10},semithick]
    \begin{scope}
      \node[state, initial above] (q0) {$q_0$};
      \node [below=1cm of q0] {$\Aa_\varphi$};
      \path[->]
      (q0) edge [double,loop right] node[label={right:$b$}] {} () 
      (q0) edge [double,loop left] node[label={left:$g$}] {} ();
    \end{scope}
        \begin{scope}[xshift=5cm,yshift=1cm]
      \node[box state,initial above, fill=goodcellcolor,label=right:$g$] (S0) {$s_0$};
      \node[prob state,fill=safecellcolor] (Prob)
           [below=1.5cm of S0] {};
           \node[box state,fill=badcellcolor,label=left:$b$] (S1)
                [left=2.5cm of Prob] {$s_1$};
      \node[prob state,fill=safecellcolor] (Prob2)
           [above=1.5cm of S1] {};
           \node [below=0.5cm of Prob] {$\Mm$};
           
      \path[-]
      (S0) edge[swap] node[label={left:$ $}] {} (Prob);
      \path[->]
      (Prob) edge[swap] node {$\frac{1}{2}$} (S1)
             edge[bend right] node[label={right:$\frac{1}{2}$}]{} (S0);
      \path[-]
      (S1) edge[swap] node[label={left:$ $}] {} (Prob2);
      \path[->]
      (Prob2) edge[swap] node {$\frac{1}{2}$} (S0)
             edge[bend right] node[label={left:$\frac{1}{2}$}]{} (S1);
    \end{scope}
        \begin{scope}[xshift=10cm,yshift=1cm]
      \node[box state,initial above, fill=goodcellcolor,label=right:$g$] (S0)
           {$s_0, q_0$};
      \node[prob state,fill=safecellcolor] (Prob)
           [below=1.5cm of S0] {};
           \node[box state,fill=badcellcolor,label=left:$b$] (S1)
                [left=2.5cm of Prob] {$s_1, q_0$};
      \node[prob state,fill=safecellcolor] (Prob2)
           [above=1.5cm of S1] {};
           \node [below=0.5cm of Prob] {$\Mm \times \Aa_\varphi$};
           
      \path[-]
      (S0) edge[double,swap] node[label={left:$ $}] {} (Prob);
      \path[->]
      (Prob) edge[swap] node {$\frac{1}{2}$} (S1)
             edge[bend right] node[label={right:$\frac{1}{2}$}]{} (S0);
      \path[-]
      (S1) edge[double,swap] node[label={left:$ $}] {} (Prob2);
      \path[->]
      (Prob2) edge[swap] node {$\frac{1}{2}$} (S0)
             edge[bend right] node[label={left:$\frac{1}{2}$}]{} (S1);
    \end{scope}
  \end{tikzpicture}
  \caption{\small {\bf Syntactic and Semantic Probabilities Agree. }
    A deterministic B\"uchi automaton $\Aa_\varphi$
    (left) that recognises the language of infinitely many $g$'s or infinitely many
    $b$'s (accepts all $\omega$-words) and an MDP $\Mm$ (center), whose set of actions is singleton (a Markov chain). 
    Again double edges mark accepting transitions.
    Notice that the MDP $\Mm$ satisfies the property with probability $1$. Their
    product is shown on the right side where the whole MDP is one accepting
    end-component. Hence, the probability of reaching the end-component (under any strategy) is $1$.} 
  \label{fig:DBW}
\end{figure}
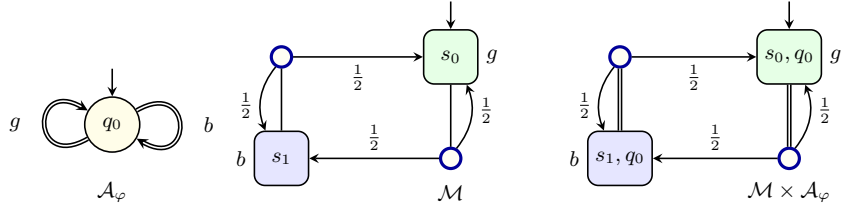

\subsection{GFM B\"uchi Automata and Reinforcement Learning}
\label{sec:rl}
The limit reachability technique~\cite{Hahn19} reduces the model checking problem for given MDP
and GFM B\"uchi automaton to a reachability problem by slightly changing the structure of the
product: One adds a target state $t$ that can be reached with a given
probability $1 - \zeta$ whenever visiting an accepting transition of
the original product MDP.
This reduction avoids the identification of accepting end-components and
thus allows a natural integration to a wide range of model-free RL
approaches.  Thus, while the proofs do lean on standard model checking
properties that are based on identifying winning end-components, they
serve as a justification not to consider them when running the
learning algorithm.

  For any $\zeta \in (0,1)$, the \emph{augmented MDP} $\Mm^\zeta$ is an MDP 
  obtained from $\Mm\!\times\!\Aa$ by adding a sink state $t$ with a self-loop
  to the set of states of $\Mm\!\times\!\Aa$, and by making $t$ a destination
  of each accepting transition
  $\tau$ of $\Mm\!\times\!\Aa$ with probability
  $1-\zeta$.  
  The original probabilities of all other destinations 
  of an accepting
  transition $\tau$ are multiplied by $\zeta$.
An example of an augmented MDP is shown in Figure~\ref{fig:add-sink}.
With a slight abuse of notation, 
if $\sigma$ is a strategy on the augmented MDP
$\Mm^\zeta$, we denote by $\sigma$ also the strategy on $\Mm \!\times\!\Aa$
obtained by removing $t$ from the domain of $\sigma$.
The following result shows the correctness of the construction.
\begin{theorem}[Limit Reachability Theorem~\cite{Hahn19}]
  \label{thm:limit-reach}
  There exists a threshold $\zeta' \in (0,1)$ such that, for all
  $\zeta > \zeta'$ and every state $s$, any strategy $\sigma$ that
  maximizes the probability of reaching the sink in $\Mm^\zeta$ is (1) an
  optimal strategy in $\Mm\times\Aa$ from $s$ and (2) induces an optimal
  strategy for the original MDP $\Mm$ from $s$ with the objective to produce a run in the language of  $\mathcal A$.
  Moreover, $\Mm$ produces such a run almost surely if, and only if, the sink is
  almost surely reachable in $\Mm^\zeta$ for all $0 < \zeta < 1$.
\end{theorem}
\noindent
Theorem~\ref{thm:limit-reach} leads to a very simple model-free RL algorithm for GFM B\"uchi automata.  
The augmented product is not built by the RL algorithm, which does not
know the transition structure 
of the environment MDP.
Instead, the observations are used to drive the objective automaton.  When the
automaton reports an accepting transition, the interpreter tosses a  biased coin
to give the learner a reward with probability $1-\zeta$. 
The interpreter also extracts the set of actions for the learner to choose from.
If the automaton is not deterministic and it has not taken the one
nondeterministic transition it needs to take yet, the set of actions the
interpreter provides to the learner includes the choice of special ``jump''
actions that instruct the automaton to move to a chosen accepting component.
When the automaton reports an accepting transition, the interpreter
gives the learner a positive reward with probability $1-\zeta$.  When
the learner actually receives a reward, the training episode
terminates.  Any RL algorithm that maximizes this probabilistic reward
is guaranteed to converge to a policy that maximizes the probability
of satisfaction of the objective.

\begin{figure}[t] 
  \centering
  \begin{tikzpicture}[scale=0.85,transform shape]
    \colorlet{darkgreen}{green!40!black}
    \begin{scope}
      \node[box state,fill=safecellcolor] (S0) {$0$};
      \node[prob state,fill=safecellcolor] (Prob)
      [below=1.5cm of S0] {};
      \node[box state,fill=safecellcolor] (S1)
      [left=2.5cm of Prob] {$1$};
      \node[box state,fill=safecellcolor] (S2)
      [right=2.5cm of Prob] {$2$};
      \path[->]
      (S2) edge[bend right=30, swap] node {$\tau_2$} (S0);
      \path[-]
      (S0) edge[double,swap] node[label={left:$\tau_0$}] {} (Prob);
      \path[->]
      (Prob) edge[swap] node {$p$} (S1)
             edge node {$1-p$} (S2)
      (S1) edge[double,bend left=30] node[label={$\tau_1$}] {} (S0);
    \end{scope}
    \node[single arrow, fill=black!55] at (3.75cm,-1.5cm) {\phantom{m}};
    \begin{scope}[xshift=8.5cm]
      \node[box state,fill=safecellcolor] (S0) {$0$};
      \node[prob state,fill=safecellcolor] (Prob)
      [below=1.5cm of S0] {};
      \node[box state,fill=safecellcolor] (S1)
      [left=2.5cm of Prob] {$1$};
      \node[box state,fill=safecellcolor] (S2)
      [right=2.5cm of Prob] {$2$};
      \path[->]
      (S2) edge[bend right=30, swap] node {$\tau_2$} (S0);
      \node[prob state,fill=safecellcolor] (Tau1)
      [left=1.2cm of S1] {};
      \node[box state, darkgreen, thick, fill=goodcellcolor] (Sink)
      [below=1.5cm of Prob] {$t$};
      \path[-]
      (S0) edge[swap] node {$\tau_0$} (Prob);
      \path[->]
      (Prob) edge[swap] node {$p\zeta$} (S1)
             edge node {$(1-p)\zeta$} (S2)
             edge[darkgreen] node[black] {$1-\zeta$} (Sink)
      (Tau1) edge[out=90,in=180] node {$\zeta$} (S0)
             edge[out=270,in=180,swap,darkgreen] node[black] {$1-\zeta$} (Sink);
     \path[-] (S1) edge[swap] node {$\tau_1$} (Tau1);
    \end{scope}
  \end{tikzpicture}
 \caption{\small Adding transitions to the target in the augmented product MDP.}
 \label{fig:add-sink}
\end{figure}
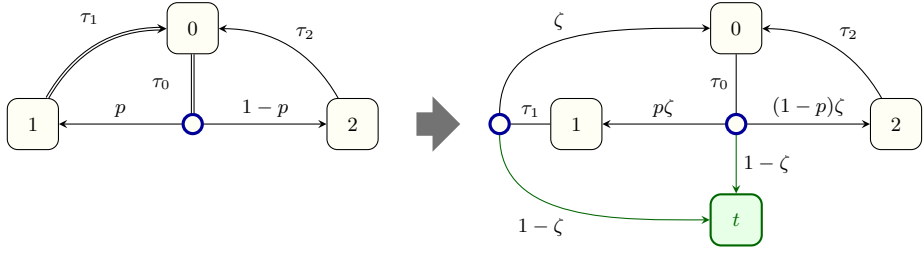

\section{Alternating GFM Automata}
\label{sec:outline}

Before giving a translation from Deterministic Streett automaton (DSA) to a good-for-MDP (GFM) alternating B\"uchi automaton (ABA), let us see how simple the translation is for its dual, a Deterministic Rabin Automaton (DRA).
When we start with a DRA $\mathcal R$, the translation to a GFM Nondeterministic B\"uchi Automaton (NBA) \cite{Hahn20}
is straightforward as shown next.
\begin{definition}[DRA to GFM NBA]
  For a given deterministic Rabin automaton $\mathcal{R} = \tuple{\Sigma,Q,q_0,\delta,\{\langle R_i,G_i\rangle \mid i\in \{1,\ldots,k\}}$, we construct a nondeterministic GFM automaton
$\mathcal{B} = \tuple{\Sigma,Q \times\{0,\ldots,k\}\cup \{ \bot\},(q_0,0),\delta',F}$ where:

\[
\delta'(\hat{q}, a) = \begin{cases}
\{\delta(q,a)\} \times \{0,\ldots,k\} & \text{ if $\hat{q} = (q, 0)$} \\
\{(\delta(q,a),i)\} & \text{ if $\hat{q} = (q, i)$, $i \not = 0$, and $q \notin R_i$} \\
\bot & \text{otherwise.}
\end{cases}
\]
and $F= \{(q,i) \mid i \in \{1,\ldots,k\} \mbox{ and } q \in G_i\}$.
\end{definition}

The resulting NBA makes only a single guess:
it guesses when an accepting end-component is reached in the product MDP $\mathcal M \times \mathcal R$ (noting that the $0$-copy is congruent to the original automaton), and then moves to a copy $i$, whose pair makes this end-component accepting.
It is easy to see that this automaton is language equivalent to $\mathcal R$ and
good-for-MDPs (e.g., it satisfies the simulation condition from \cite{Hahn20}).
For $k$ pairs, this creates only $k{+}1$ copies, and thus a small overhead; and it allows one to then use standard reward translation techniques for B\"uchi acceptance conditions \cite{Hahn19} in RL.

The question of how to maximise the probability to satisfy a Streett condition (or, likewise, how to minimise the probability to satisfy a Rabin condition) is more challenging.
Broadly speaking, the translation of a Rabin acceptance condition is simplified by the fact that the nondeterministic choices of an NBA can easily handle the resolution of the disjunction of the acceptance condition on pairs, and resolving nondeterminism is something that always needs to be done when analysing an MDP.
However, it is harder to accommodate for a conjunction of the acceptance condition on pairs, as in a Streett acceptance condition.
As a consequence, the translation of a deterministic Streett automaton to a nondeterministic Rabin automaton (without a restriction to GFM) leads to a blow-up that results in $2^{\theta(n)}$ states \cite{SV89}, while a translation to an NBA requires $n2^{\theta(k)}$ states \cite{SV89}, even without the restriction to GFM.

\emph{Surprisingly, there is a way to exploit \emph{alternating} good-for-MDP automata with a small blow-up of $k{+}2$ for $k$ pairs.}

As B\"uchi games can be handled with similar techniques as for B\"uchi MDPs (Section~\ref{sec:rl}) in model-free reinforcement learning (cf.~\cite{Hahn20b}), the alternation itself does not create problems during learning; still, it is quite surprising that such a method works. 
This is partly because of the exponential memory requirement for Streett conditions, and partly because the acceptance player for the MDP would not have access to decisions the rejection player has made in the resulting game.
However, while the automaton is small, the memory we infer from the winning strategy of this small automaton can be exponentially larger.

An optimal strategy in the resulting game does not in itself constitute a strategy for controlling the MDP for a given DSA.
This is because different strategic choices of the antagonistic rejection player will lead to different positions in the game, and there is no guarantee that a consistent positional strategy for all of these positions exists.
Moreover, the strategic choices of an antagonistic rejection player have no direct relation to the observable history.
We show, however, that the history can be used to identify a state in the game, whose decisions the acceptance player should follow.

The need for memory is, therefore, not gone.
Instead, the control strategy we construct in the correctness proof for the resulting alternating GFM B\"uchi automaton is only one part of the control strategy used for the MDP.
The other is a latest appearance record (LAR), which is kept in addition to the constructed game.
The LAR will determine the state, from a family of equivalent states, whose strategy will be followed.

\subsection{Alternating GFM Automata}
\label{sec:agfm}
There is a number of mildly different definitions of alternating automata, and we can use the simplest one, where the states are partitioned into nondeterministic and universal states.

\begin{definition}
  An alternating $\omega$-automaton $\mathcal{A} =
  \tuple{\Sigma,Q_n,Q_u,q_0,\delta,\alpha}$, with $Q=Q_n\cup Q_u$, is an automaton such that $\tuple{\Sigma,Q,q_0,\delta,\alpha}$ is a nondeterministic automaton, and $Q_n$ and $Q_u$ are disjoint sets of nondeterministic and universal states, respectively.
\end{definition}

  A \emph{run tree} of an alternating automaton
  $\mathcal{A} = \tuple{\Sigma,Q_n,Q_u,q_0,\delta,\alpha}$ on word
  $w \in \Sigma^\omega$ is a family of functions $\{\rho_j : \omega \to Q \mid j \in J\}$ for some non-empty index set $J$ such that 
  \begin{itemize}
      \item  $\rho_j$ is a run for all $j \in J$, and
      \item if $\rho_j$ has a universal state $q'$ at a position $i \in \omega$ ($\rho_j(i) = q' {\in} Q_u$), then, for all $q \in \delta\big(q',w(i)\big)$, there is a $j_q \in J$ such that $\rho_{j_q}(k)=\rho_{j}(k)$ for all $k \leq i$, and $\rho_{j_q}(i+1)=q$.
  \end{itemize}
    A run tree is accepting if all of the runs of $\{\rho_j : \omega \to Q \mid j \in J\}$ are accepting.

A minimal such family of runs can be viewed as a tree, where nondeterministc states have one successor, while universal states have many, namely all those defined by the local successor function.
Alternatively, a family of runs can be viewed as a game, where an angelic acceptance player chooses the successor for a nondeterministic state, while an antagonistic rejection player selects the successor for a universal state.
This way, they successively construct a run, and acceptance is decided by whether or not this run accepts.

We extend the product construction from Section \ref{ssec:product} to produce a B\"uchi game from the product $\Mm {\times} \Aa$ of an MDP $\Mm$ with an alternating B\"uchi automaton $\Aa$, where the decisions of the rejection player are simply the decision to resolve the nondeterminism from the universal states, while resolving the nondeterminism from the MDP \emph{and} resolving the nondeterminism from the nondeterministic automata states are left to the acceptance player.
Both players have positional optimal strategies (where, for the rejection player, positionality in\-cludes the state and the choice made by the acceptance player) in this game~\cite{mciver2002games}.

We refer to the probability, with which the acceptance player can win this product game from a product state $(s,q)$ as 
\begin{align*}
  \PSat^\Mm_{\Aa}(s,q) &= \sup_{\sigma \in \Sigma_{\Mm \times \Aa}} \inf_{\pi \in \Pi_{\Mm \times \Aa}} \eE^{\Mm\times\Aa}_{\sigma,\pi}(s,q) \set{ \alpha(\seq{X_0, X_1, \ldots}) } \; ,
\end{align*}
where $\alpha$ is the B\"uchi condition and $X_i$ is the random variable corresponding to the state of the automaton at the $i$-th step.

\begin{definition}[Alternating GFM Automata]
  \label{def:gfm}
An alternating automaton $\Aa$ is \emph{good for MDPs} if, for all
  MDPs $\Mm$, $\PSat^\Mm_{\Aa}(s_0,q_0) = \PSemSat^\Mm_{\mathcal A}(s_0)$ holds,
  where $s_0$ is the initial state of $\Mm$.
\end{definition}

\subsection{Construction of the Alternating B\"uchi Automaton}
\label{ssec:construction}
The motivation for the translation of a deterministic Streett automaton to a GFM automaton is similar to that for Rabin:
when having nondeterministic power, we can use it to guess when we have reached an accepting end-component that we plan to cover completely (i.e., we will almost surely visit every state and every transition in the end-component infinitely often).

While covering an accepting end-component may require memory (or randomisation), its properties with respect to the Streett condition are straightforward:
for every Streett pair $\langle G,R \rangle$, if the end-component contains a red state $q \in R$ then it must also contain a green state $q' \in G$ from the same pair, which should (almost surely) be visited after every visit of $q$.

\begin{definition}[DSA to Alternating GFM B\"uchi]
For a given deterministic Streett automaton $\mathcal{S} = \tuple{\Sigma,Q,q_0,\delta,\{\langle G_i,R_i\rangle \mid i\in \{1,\ldots,k\}}$, where we assume without loss of generality that $G_i\cap R_i = \emptyset$ for all $i=1,\ldots,k$, we construct an alternating B\"uchi automaton
$\mathcal{A} = \tuple{\Sigma,Q,Q\times\{0,\ldots,k\},q_0,\delta',F}$ where:
\begin{itemize}
\item 
First, for every state $q \in Q$, we let $I_q = \{0\} \cup \{i \mid q \in R_i\}$.

\item
We now define, for every state $q \in Q$ and letter $a \in \Sigma$, where $q' = \delta(q,a)$:
\begin{itemize}
    \item 
$\delta'(q,a)= \{q',(q',0)\}$ and,

\item for all $i = 0,\ldots,k$, 
$\delta'((q,i),a) = \{q'\}\times (I_{q'}\smallsetminus \{i\})$ if $q' \in G_i$ and \\
$\delta'((q,i),a) = \{q'\}\times I_{q'}$ if $q' \notin G_i$.
\end{itemize}
\item Finally, we set the set of final transitions to $F = \{(q,i),a,(q',j) \mid i\neq j \mbox{ or } i=j=0\}$.
\end{itemize}
\end{definition}
Note that the projection on the state of $\mathcal S$ is not affected by this translation.

The intuition for this translation is that the acceptance game starts in the original copy of the states---the nondeterministic states $Q$.
From there, the acceptance player can \emph{declare} when he has reached an accepting end-component, moving from the original copy to the $0$-copy of the game.
The rejection player can henceforth, whenever a state from the $i^{th}$ red set $R_i$ is seen, move from a $j$-copy 
to the $i$-copy, which can be viewed as a claim that the requirement on the $i^{th}$ Streett pair is not fulfilled (finitely many $R_i$ or infinitely many $G_i$ states).
She therefore \emph{challenges} the acceptance player to visit a state from the $i^{th}$ green set $G_i$ (an $i$-challenge for short).
When the game is in the $j$-copy, the game moves back to the $0$-copy when no new challenge is made and a state in $G_j$ is visited.
Otherwise, the game stays in the $j$-copy.

The acceptance player wins if the rejection player makes infinitely many challenges (the $i \neq j$ part of the final transitions) or if the game stays infinitely often in the $0$-copy (the $i=j=0$ part of the final states).
The rejection player wins if
the acceptance player never declares, or if she makes only finitely many challenges, and her last challenge is never met.

To keep the definition simple, we have allowed the rejection player to always withdraw a challenge by moving back to the $0$ copy without reason.
This is never an attractive move for her (so long as she has other options), and can therefore be omitted in an implementation.

\begin{figure}[t]
\centering
\begin{tikzpicture}[
every text node part/.style={align=center},
every state/.style={fill=safecellcolor}]
\node[state,initial] (q0) at (0,0) {$q_0$};
\node[state] (q1) at (2,0) {$q_1$};
\draw[->] (q0) edge[loop above] node[above] {$a$} ();
\draw[->] (q0) edge[bend left=20] node[above] {$b$} (q1);
\draw[->] (q1) edge[loop above] node[above] {$b$} ();
\draw[->] (q1) edge[bend left=20] node[below] {$a$} (q0);
\node at (5,0.25) {$G_1 = \{q_0\}, R_1 = \{q_1\}$\\[1mm]{}$G_2 = \{q_1\}, R_2 = \{q_0\}$};
\end{tikzpicture}
\caption{Streett automaton on the alphabet $\Sigma = \{a,b\}$.}
\label{fig:streett}
\end{figure}
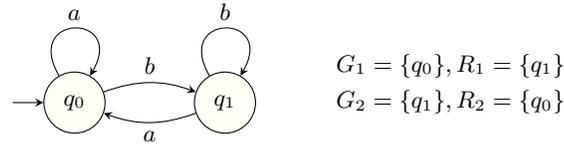

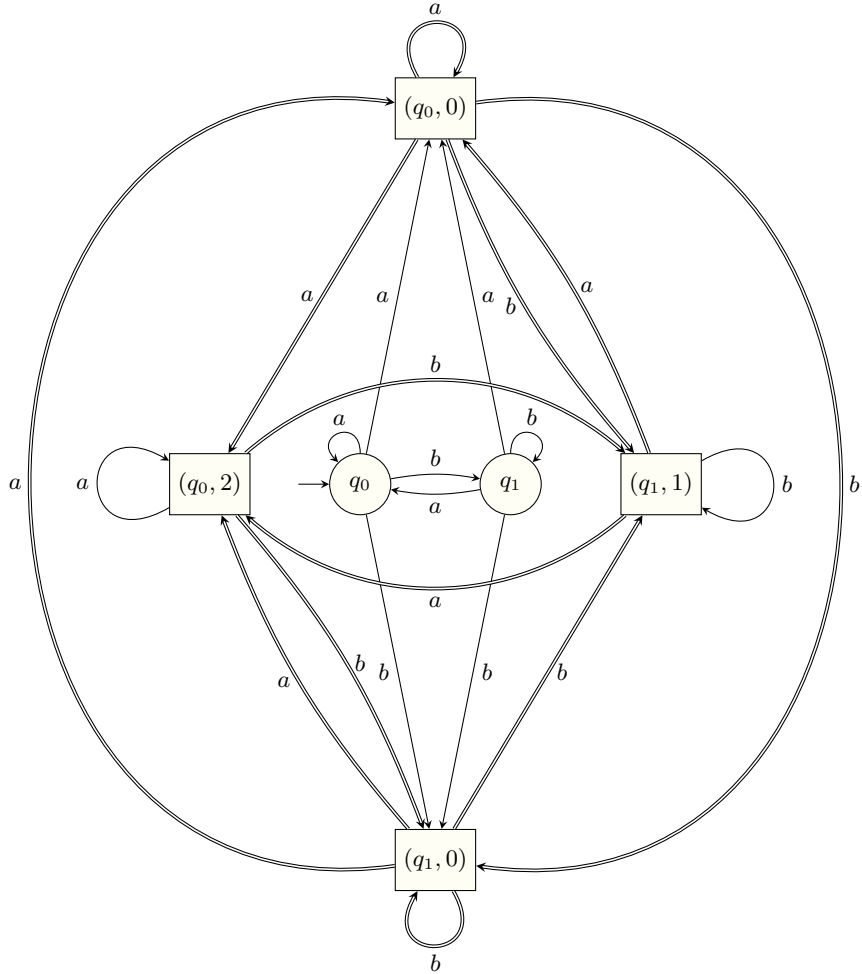
\begin{figure}[t]
\centering
\begin{tikzpicture}[
  every text node part/.style={align=center},
  every state/.style={fill=safecellcolor}]
\useasboundingbox (-5,-6.6) rectangle (7,6.6);
\node[state,initial] (q0) at (0,0) {$q_0$};
\node[state] (q1) at (2,0) {$q_1$};
\node[state,rectangle] (q00) at (1, 5) {$(q_0,0)$};
\node[state,rectangle] (q10) at (1,-5) {$(q_1,0)$};
\node[state,rectangle] (q01) at (-2, 0) {$(q_0,2)$};
\node[state,rectangle] (q12) at (4, 0) {$(q_1,1)$};
\draw[->,every loop/.style={looseness=3}] (q0) edge[in=135,out=90,loop] node[above] {$a$} ();
\draw[->] (q0) edge[bend left=10] node[above] {$b$} (q1);
\draw[->] (q1) edge[bend left=10] node[below] {$a$} (q0);
\draw[->,every loop/.style={looseness=3}] (q1) edge[in=45,out=90,loop] node[above] {$b$} ();
\draw[->] (q0) -- node[left] {$a$} (q00);
\draw[->] (q0) -- node[left] {$b$} (q10);
\draw[->] (q1) -- node[right] {$a$} (q00);
\draw[->] (q1) -- node[right] {$b$} (q10);
\draw[->,double] (q00) -- node[left] {$a$} (q01);
\draw[->] (q00) edge[loop above,double] node[above] {$a$} ();
\draw[->,double] (q00) .. controls +(7,1) and +(7,-1) .. node[right] {$b$} (q10);
\draw[->] (q00) edge[bend right=10,double] node[left] {$b$} (q12);
\draw[->,double] (q10) edge[loop below,double] node[below] {$b$} ();
\draw[->,double] (q10) .. controls +(-7,-1) and +(-7,1) .. node[left] {$a$} (q00);
\draw[->] (q01) edge[loop left] node[left] {$a$} ();
\draw[->] (q01) edge[bend left=10,double] node[right] {$b$} (q10);
\draw[->] (q01) edge[bend left=41,double] node[above] {$b$} (q12);
\draw[->] (q12) edge[bend left=41,double] node[below] {$a$} (q01);
\draw[->] (q10) edge[bend left=10,double] node[left] {$a$} (q01);
\draw[->,double] (q10) -- node[right] {$b$} (q12);
\draw[->] (q12) edge[loop right] node[right] {$b$} ();
\draw[->] (q12) edge[bend right=10,double] node[right] {$a$} (q00);
\end{tikzpicture}
\caption{Alternating GFM corresponding to the deterministic Streett automaton shown in Figure~\ref{fig:streett}. Here double edges depict accepting transitions.  The states controlled by the acceptance player are shown as circles, while the states controlled by the rejection player are shown as boxes.}
\label{fig:streett-to-gfm}
\end{figure}

\begin{example}
Consider the deterministic Streett automaton sketched in Figure \ref{fig:streett}.
The language of this automaton is $((a|b)^*ab)^\omega$ (seeing infinitely many $a$-s and infinitely many $b$-s):
because of its acceptance condition, from $(G_1,R_1)$, we must infinitely often see $q_0$ or only finitely see $q_1$ and at the same time from $(G_2,R_2)$ we must infinitely often see $q_1$ or only finitely see $q_0$.
It is easy to see that this condition is equivalent to requiring that we infinitely often see both $q_0$ and $q_1$.
Therefore, we require to see infinitely many $a$-s and infinitely many $b$-s.

In Figure~\ref{fig:streett-to-gfm}, we provide a translation to an alternating GFM B\"uchi automaton, where
double arrows indicate accepting transitions.
The resulting GFM automaton contains the original Streett automaton, i.e., $q_0$, $q_1$, and the transitions between them. They are the states referred to as ``the original copy,'' and the only nondeterministic states (states where the acceptance player chooses---depicted as circles).
None of the transitions from the original copy are accepting.

At some point, the acceptance player can make the decision to move to the final part of the automaton, consisting of states of the form $(q,i)$ (which is referred to as the acceptance player declaring).
All of these states are universal: only the rejection player makes choices. Universal states are depicted as squares.

In this part of the automaton, all transitions except the loops on $(q_0,2)$ and $(q_1,1)$ are accepting.
Because of this, if the last character read is an $a$, the rejection player can move to $(q_0,2)$ and stay in the non-accepting loop until a $b$ is read.
Similarly, if the last letter read is an $a$, the rejection player can move to $(q_0,2)$ and stay in the non-accepting loop until a $b$ is read.
Likewise, when the last letter read is a $b$, the rejection player can move to $(q_1,1)$ and stay in the non-accepting loop until a $b$ is read.
Thus, the alternating B\"uchi automaton recognises the same language as the original deterministic Streett automaton.

As we have remarked in the construction, voluntary moves of the rejection player to the $0$-copy (intuitively: withdrawing her latest challenge) were allowed only to simplify the definition and can be omitted in an implementation. They are not included in the drawing to avoid clutter.
\end{example}

\section{Correctness of the Construction}
\label{sec:correct}

In order to prove that the alternating B\"uchi automaton is good-for-MDPs, we first show that using this automaton provides at least the same syntactic probability to win as using the deterministic Streett automaton $\mathcal S$.

\begin{lemma}
\label{lem:atMostSame}
Let $\mathcal S$ be a deterministic Streett automaton and $\mathcal A$ the alternating automaton from above constructed from $\mathcal S$.
Then, for every MDP $\mathcal M$, $\mathcal M \times \mathcal A$ has at least the same winning probability as $\mathcal M \times \mathcal S$.
\end{lemma}

\begin{proof}
We first observe that the acceptance player (as the Streett player in a finite state Streett game) has an optimal pure finite state strategy $\sigma$ for $\mathcal M \times \mathcal S$.
Let $(\mathcal M \times \mathcal S)_\sigma$ be the Markov chain obtained by using this optimal control.

In $(\mathcal M \times \mathcal S)_\sigma$, we will almost surely reach a leaf component, and the chance of winning is the chance of reaching an accepting leaf component (i.e., a leaf component where the Streett condition is almost surely satisfied).

For $\mathcal M \times \mathcal A$, we now define a pure finite state strategy $\tau$ for the acceptance player from $\sigma$ and $(\mathcal M \times \mathcal S)_\sigma$ as follows.
Outside the accepting leaf components, we follow $\sigma$ and stay in the original copy.
When entering an accepting leaf component, we move to the $0$-copy, but otherwise make the same decision as for $\sigma$. Henceforth, we make the same decision that $\sigma$ would make on the history obtained by ignoring in which $i$-copy we are.
(Note that the decision on making an $i$-challenge, and hence on which $i$-copy should be visited, rests with the rejection player.)

As this was an accepting leaf component in $(\mathcal M \times \mathcal S)_\sigma$, if there is, for any pair $\langle G_i,R_i \rangle$, a (red) state in $R_i$ in the leaf component, there is also a (green) state in $G_i$, and this state is almost surely visited infinitely often.
Consequently, every challenge will, almost surely, eventually be met, and the acceptance player will win almost surely from these positions, regardless of how the rejection player plays.
Thus, $\tau$ provides (at least) the same probability to win in $\mathcal M \times \mathcal A$ as $\sigma$ provides for $\mathcal M \times \mathcal S$.
\qed
\end{proof}

Different to the case of nondeterministic good-for-MDP automata originally suggested in \cite{Hahn20}, we also have to show that the probability of winning for $\mathcal A$ cannot exceed that for $\mathcal S$.

\begin{lemma}
\label{lem:streett-at-least-same-prob}
Let $\mathcal S$ be a deterministic Streett automaton and $\mathcal A$ the alternating automaton from above constructed from $\mathcal S$.
Then, for every MDP $\mathcal M$, $\mathcal M \times \mathcal S$ has at least the same winning probability as $\mathcal M \times \mathcal A$.
\end{lemma}

Before starting the proof, we define useful terminology, and make the assumption, for simplicity, that a positional optimal strategy $\sigma$ for the acceptance player on $\mathcal M \times \mathcal A$ has been fixed.

We call two states of
$\mathcal M \times \mathcal A$ \emph{related}, if they refer to the same vertex of $\mathcal M$ and $\mathcal S$, but possibly to different copies of this state in $\mathcal A$.
For such related states, it is obviously the case that the probability to win from the $0$-copy is at least as high as the probability to win from any other $i$-copy, as the acceptance player can just play as if he started in that $i$-copy until the time where the first challenge is made.
(The only difference with respect to acceptance from the $i$-copy is then that paths where no challenge is made become winning, such that the probability to win can only go up.)
We further observe that the probability to win from the original copy is always at least as high as the probability to win from the $0$-copy, as the acceptance player can always declare.

We therefore coin the term ``good copy'' of a state: a copy of a state is \emph{good} if, and only if, the probability of winning from this copy is as high as the probability of winning from the original copy.
A good copy is called \emph{reachable} if it is reachable in $(\mathcal M \times \mathcal A)_\sigma$.
The \emph{oldest} reachable good copy of a state is the good copy $i$, for which the last visit to $G_i$ is longest ago, where the higher number is given preference as a tie breaker. In particular, the $0$-copy is only the oldest reachable copy, when it is the only reachable good copy different to the original copy. If no other reachable copy is good, the original copy is the oldest reachable good copy.
Naturally, all $\sigma$-successors of a reachable good copy are reachable good copies.

Note that the property of being the oldest reachable good copy is relative to the history; a latest appearance record (also known as index appearance record) \cite{Gurevi82,DBLP:conf/stoc/Safra92,DBLP:journals/siamcomp/Safra06,thomas03}
is a standard memory structure of size $k!$ for keeping track of all information required for determining the oldest copy for a given history.
Let $M_{\mathcal S}$ be such a memory structure.

\begin{proof}
Let $\sigma$ be an optimal positional strategy of the acceptance player in the B\"uchi game $\mathcal M \times \mathcal A$, and let $\mathcal S' = \mathcal S \times M_{\mathcal S}$ be $\mathcal S$ equipped with a latest appearance record with $>$ as a tie breaker.
We use this to construct the positional strategy $\tau$ for $\mathcal M \times \mathcal S'$ as the strategy that makes the same choice $\sigma$ makes for the oldest reachable good copy of that state in the $\mathcal S$ projection of $\mathcal S'$.

It now suffices to show that the rejecting leaf components of $(\mathcal M \times \mathcal S')_\tau$ refer to states of $\mathcal M \times \mathcal A$, whose good copies have a winning probability of $0$.

We first assume that there is a reachable leaf component that contains a state, where the oldest reachable good copy is the original copy.
Note that this implies that the original copy is the only reachable good copy of that state.
Naturally, the successor of a reachable good copy under $\sigma$ is a reachable good copy, so every predecessor of the original copy, and by induction the complete leaf component, consists of states, where the original copy is the only good reachable copy.
Thus, this leaf component in $(\mathcal M \times \mathcal S)_\tau$ projects into an end-component in
$(\mathcal M \times \mathcal A)_\sigma$, where the rejection player has no decisions, and where no final transition occurs.
The winning probability of all states in this end-component is $0$.

We now assume that the rejecting leaf component contains only states with the same oldest reachable copy $i \geq 1$.
Then the leaf component follows the positional strategy for the $i$-copy in $(\mathcal M \times \mathcal A)_\sigma$; note that this entails that it does not contain a state in $G_i$.
Therefore the rejection player surely wins in the $i$-copy of this end-component in $(\mathcal M \times \mathcal A)_\sigma$ by never changing her challenge.

Let us finally turn to the case where a leaf component in $(\mathcal M \times \mathcal S)_\tau$ contains only states, where all oldest reachable good copies are not the original copy, and that these copies are different, or all $0$.
We assume for contradiction that the leaf component is rejecting.
Then there must be an index $i$ such that there is a (red) state from $R_i$ in the leaf component, but not a (green) state from $G_i$. Moreover, there must be an $i^*$ with this property where, in the given history, the last occurrence of $G_{i^*}$ is longest ago, using $>$ as tie breaker.
Further, let us consider a path through this leaf component that visits states from all (green) sets $G_{i'}$ represented in this leaf component.

Let us now consider a (red) state in $R_{i^*}$ in the leaf component.
If the $j$-copy is not the $i^*$ copy, then, as the rejection player can make an $i^*$ challenge, the $i^*$ copy (as a viable successor under the optimal strategy) must be a reachable good copy of the state, too, and therefore, by our assumption, the oldest reachable good state.
Thus, we move on to the $i^*$ copy, and henceforth never leave it, contradicting the assumption that we are in a leaf component that contains different copies, or only the $0$-copy, as oldest reachable states.

We have shown that we almost surely reach a leaf component, where the probability of winning all related states is $0$ in $\mathcal M {\times} \mathcal A$, or where the chance of winning is $1$.
Together with the local consistency of the probabilities, we get the claim.
\qed
\end{proof}

The two lemmas from this section imply that the syntactic and semantic probability to win are the same for all MDPs---in short, that $\mathcal A$ is good-for-MDPs.
This in particular implies language equivalence on ultimately periodic words (which are a special case of Markov chains, where every state has only one successor), and therefore on all words, as two $\omega$-automata that accept the same ultimately periodic words recognise the same language.

Moreover, we have provided a translation of an optimal strategy obtained for $\mathcal M \times \mathcal A$ into a strategy for $\mathcal M \times \mathcal S$ with (at least, and then with Lemma \ref{lem:atMostSame} precisely) the same optimal probability to win in the proof of Lemma \ref{lem:streett-at-least-same-prob}.

\begin{corollary}
\label{cor:correct}
The alternating B\"uchi automaton $\mathcal A$ that results from the construction of Section \ref{ssec:construction} from a DSA $\mathcal S$ is a good-for-MDPs automaton that recognises the same language as $\mathcal S$.
Moreover, we can infer an optimal control strategy for the acceptance player for $\mathcal M \times \mathcal S$ from an optimal strategy of the acceptance player in $\mathcal M \times \mathcal A$.
\qed
\end{corollary}

\begin{figure}[t]
\centering
\begin{tikzpicture}[
every text node part/.style={align=center},
every state/.style={fill=safecellcolor}]
\node[state,initial,inner sep=1pt] (q0) at (0,0) {$q_0,(1,2)$};
\node[state,inner sep=1pt] (q1) at (3.5,0) {$q_1,(2,1)$};
\draw[->] (q0) edge[loop above] node[above] {$a$} ();
\draw[->] (q0) edge[bend left=20] node[above] {$b$} (q1);
\draw[->] (q1) edge[loop above] node[above] {$b$} ();
\draw[->] (q1) edge[bend left=20] node[below] {$a$} (q0);
\node at (7,0.25) {$G_1 = \{q_0\}, R_1 = \{q_1\}$\\[1mm]{}$G_2 = \{q_1\}, R_2 = \{q_0\}$};
\end{tikzpicture}
\caption{Streett automaton of Fig.~\ref{fig:streett} extended with latest appearance record as memory.}
\label{fig:streett-memory}
\end{figure}
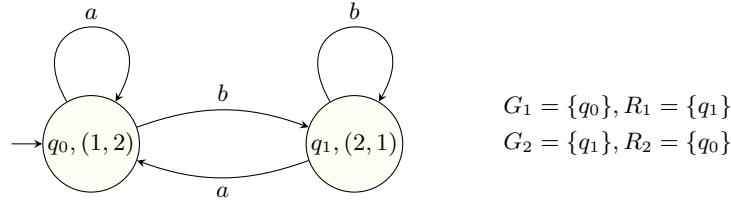

\paragraph{\bf Example.}

In Figure~\ref{fig:streett-memory}, we have extended the Streett automaton of Figure~\ref{fig:streett} with latest appearance record memory~\cite{DBLP:conf/stoc/Safra92,DBLP:journals/siamcomp/Safra06,thomas03} as discussed in Lemma~\ref{lem:streett-at-least-same-prob}.
The memory is added by extending the states with a vector of a permutation of the indices, such that the first entry corresponds to the index $i$ for which $G_i$ was most recently visited, the next one corresponds to the one before, etc., until the last ($k^{th}$) one, which corresponds to the oldest $G_i$ visited.
Note that, although the number of states stays the same, there are $k!$ permutations, and they could all be reachable.
This can lead---though not in this example---to a significant increase in the number of states needed to store the strategy explicitly.
\bigskip

We note that the memory we actually need is often smaller than the LAR we have mentioned, as the order can be mangled finitely often.
That would, for example, allow us to only keep the order in some SCCs, namely those where we might get stuck in (with probability $\neq 0$)---and, of course, only for those indices that occur in states within these SCCs.

Note that the definition relative to reachability under $\sigma$ is not required for correctness, but it provides the required connection to learning: when learning an optimal strategy in the game, the bit that is reachable under the optimal strategy we have learned is enough for constructing a pure finite state strategy.

\subsection{Succinctness}

Corollary \ref{sec:correct} shows that the alternating B\"uchi automaton $\mathcal A$ that results from the construction of Section \ref{ssec:construction} from a DSA $\mathcal S$ is a good-for-MDPs automaton that recognises the same language as $\mathcal S$, and the number of states of $\mathcal A$ is merely $O(kn)$, where $n$ and $k$ are the number of states and Streett pairs of $\mathcal S$.
At the same time, the translation of a deterministic Streett automaton to a nondeterministic Rabin automaton (without a restriction to GFM) leads to a blow-up that results in $2^{\theta(n)}$ states \cite{SV89}, while a translation to an NBA requires $n2^{\theta(k)}$ states \cite{SV89}, even without the restriction to GFM.

This immediately provides the following theorem.

\begin{theorem}
Alternating GFM B\"uchi automata can be exponentially more succinct than (general) nondeterministic B\"uchi and Rabin automata.
\qed
\end{theorem}

\subsection{$\mathcal A$ is not Good-for-Games}
\label{app:not4games}

We have shown that the alternating B\"uchi automaton $\mathcal A$ we have constructed from a deterministic Streett automaton $\mathcal S$ in Section \ref{ssec:construction} is good-for-MDPs.
To outline the difference, we now discuss why $\mathcal A$ is not, in general, good-for-games \cite{Henzin06} on the example of a deterministic Streett automaton $\mathcal S$ with two states, $a$ and $b$, and one Streett pair $\langle G_1,R_1 \rangle = \langle \emptyset,\{a\}\rangle$.
The automaton is in state $a$ after reading an $a$, and in state $b$ otherwise.
It recognises the language of all words that contain only finitely many $a$-s.

A counter-example, which shows that the alternating B\"uchi automaton $\mathcal A$ that results from the construction in Section \ref{ssec:construction} is not good for games \cite{Henzin06}, must have choice for the rejection player, as it is otherwise representable as an MDP.
Consider the one player game with two states, where the rejection player can choose in the initial state to play a $b$ and stay, or to play an $a$ and move on to the second state.
From this second state, the only available action is to stay and play a $b$.
In this one-player game, the rejection player can play only a single $a$.
Therefore, whatever she does, the word she constructs is in the language of the automaton.
She can, however, win the acceptance game played on the product of this two-state game and $\mathcal A$ by staying in the initial state until the acceptance player declares, and then moving on to the second state.
Following this approach, she moves to the $1$-copy for $\mathcal A$, creating an obligation to see a state from $G_1$ (which is empty) when the acceptance player declares.

She wins this game, irrespective of whether the acceptance player eventually declares (in which case the obligation she creates is never met), or never declares (as the original copy does not contain final states / transitions).
  The difference is that, as the rejection player can move in the game, the acceptance player cannot reach an end-component that he can almost surely cover.

\section{Discussion}
\label{sec:conclude}
When $\omega$-regular objectives were first used in model checking MDPs, deterministic Rabin automata were used to represent the objectives.
The same has been attempted by the reinforcement learning community:
when they first turned to $\omega$-regular objectives, they tried the tested route through deterministic Rabin automata \cite{Sadigh14}, but that translation fails as shown in \cite{Hahn19}.
Of course, with the current state of knowledge of good-for-MDPs automata, it is not hard to translate deterministic Rabin automata to nondeterministic B\"uchi automata that are good-for-MDPs, and then to analyse the product of such a B\"uchi automaton and the MDP in question.

While MDPs with B\"uchi conditions are a (relatively) easy target for RL methods (like $Q$-learning \cite{Hahn19,Hahn20}), a similar translation of Streett automata (or for minimising the chance of meeting a Rabin objective) appears prohibitive. This is because \emph{every} translation from DSAs to nondeterministic B\"uchi (or even to Rabin) automata incurs an exponential blow-up in the worst case.
Surprisingly, we found a way to allow even this accepting condition to be efficiently used in reinforcement learning by generalising the property of being good-for-MDPs to alternating automata, and by constructing an equivalent good-for-MDPs alternating B\"uchi automaton with linear overhead.

\bibliography{papers}

\end{document}